\documentclass[letterpaper, 10 pt, conference]{ieeeconf}
\IEEEoverridecommandlockouts  

\usepackage[utf8]{inputenc} 
\usepackage[T1]{fontenc}    
\usepackage{url}            
\usepackage{amsfonts}       
\usepackage{amsmath} 
\usepackage{amssymb}  
\usepackage{mathrsfs}
\usepackage[ruled]{algorithm2e}
\usepackage{bm}
\usepackage{stfloats}
\usepackage{subfigure}
\usepackage{graphicx}
\usepackage{threeparttable}
\usepackage{nicefrac}       
\usepackage{microtype}      
\usepackage[colorlinks=true, allcolors=black]{hyperref}
\usepackage[normalem]{ulem}

\newcommand{\mC}{\mathcal{C}}
\newcommand{\tmC}{\tilde{\mathcal{C}}}
\newcommand{\mD}{\mathcal{D}}
\newcommand{\mS}{\mathcal{S}}

\newcommand{\mX}{\mathcal{X}}
\newcommand{\mK}{\mathcal{K}}

\newcommand{\R}{\mathbb{R}}
\newcommand{\mbE}{\mathbb{E}}
\newcommand{\mbP}{\mathbb{P}}

\newcommand{\mO}{\mathcal{O}}

\newcommand{\defeq}{:=}

\newcommand{\innerp}[1]{\langle #1 \rangle}

\newcommand{\prob}[1]{\mathbb{P}\left( #1 \right)}

\newcommand{\ballone}[1]{\mathcal{B}^1\left(#1\right)}

\renewcommand{\top}{\mathsf{T}}

\newtheorem{definition}{Definition}[section]

\newtheorem{thm}{Theorem}
\newtheorem{proposition}{Proposition}
\newtheorem{assumption}{Assumption}
\newtheorem{remark}{Remark}[section]
\newtheorem{problem}{Problem}

\title{\LARGE \bf
Safety Verification of Stochastic Systems: A Set-Erosion Approach 
}
\author{Zishun Liu$^{1}$, Saber Jafarpour$^{2}$ and Yongxin Chen$^{1}$
\thanks{$^{1}$Zishun Liu and Yongxin Chen are with Georgia Institute of Technology, Atlanta, GA 30332 
        {\tt\small \{zliu910\}\{yongchen\}@gatech.edu}}%
\thanks{$^{2}$ Saber Jafarpour is with University of Colorado Boulder, Boulder, CO 80309
        {\tt\small Saber.Jafarpour@colorado.edu}}%
}

\begin{document}
\maketitle
\thispagestyle{empty}
\pagestyle{empty}

\begin{abstract}
We study the safety verification problem for discrete-time stochastic systems. 
We propose an approach for safety verification termed set-erosion strategy that verifies the safety of a stochastic system on a safe set through the safety of its associated deterministic system on an eroded subset.
%
The amount of erosion is captured by the probabilistic bound on the distance between stochastic trajectories and their associated deterministic counterpart. 
Building on our recent work \cite{szy2024Auto}, we establish a sharp probabilistic bound on this distance. 
%
Combining this bound with the set-erosion strategy, we establish a general framework for the safety verification of stochastic systems.
%
Our method is flexible and can work effectively with any deterministic safety verification techniques. 
We exemplify our method by incorporating barrier functions designed for deterministic safety verification, obtaining barrier certificates much tighter than existing results. Numerical experiments are conducted to demonstrate the efficacy and superiority of our method.


\end{abstract}


\section{Introduction}
\label{sec:introduction}

Safety is a fundamental requirement for a wide range of real-world systems, including autonomous vehicles, robots, power grids, and beyond. Motivated by the significance of safety, research on safety verification has flourished in recent decades. Typically, safety verification refers to the process of verifying whether the system state remains within a defined safe region over a specified time horizon, whether in discrete-time or continuous-time contexts \cite{li2023survey}. 
In this paper, we focus on the safety verification problem for discrete-time systems.

Since the safety of real-world systems is frequently challenged by uncertainties in the environment \cite{Zhang2016Understanding}, it is essential for safety verification schemes to account for disturbances.
Most existing approaches have modeled disturbances as bounded deterministic inputs and verified the safety in the worst case through deterministic methods such as dynamic programming \cite{AA-MP-JL-SS:08,SS-JL:10}, barrier certification \cite{prajna2007framework} and ISSf-barrier function \cite{8405547}. Among these deterministic methods, barrier certification has attracted growing attention thanks to its simplicity and has been widely adopted
to formally prove the safety of nonlinear and hybrid systems \cite{ames2019control}.

Many real-world applications are subject to stochastic disturbances \cite{chapman2021risk}. In such cases, traditional deterministic methods often become either inapplicable or overly conservative, as they focus on worst-case scenarios that rarely happen~\cite{cosner2024bounding}.
To better reflect the effects of stochastic disturbances, stochastic safety verification shifts the focus to ensuring safety within a safe set with high probability, e.g., a finite-time stochastic trajectory stays in the safe set with probability $>99.9\%$. 

Multiple techniques have been developed for the safety verification of discrete-time stochastic systems.
For instance, martingale-based strategies \cite{cosner2024bounding,nishimura2024control,santoyo2021barrier}  focus on constructing barrier functions that utilize semi-martingale or $c$-martingale conditions \cite{steinhardt2012finite} to bound the failure probability. Another commonly used method is direct risk estimation \cite{frey2020collision,blackmore2009convex,ono2015chance}, which first bounds the failure probability of the system state at a single time instance, then applies a union bound over the entire time horizon. Some other methods such as conformal prediction \cite{vlahakis2024probabilistic} and optimization-based approaches with chance constraints \cite{5970128} are also applied in practice. However, all these techniques are often either overly conservative for ensuring safety with high probability or limited to specific, restrictive scenarios. 

In this work, we present a novel approach termed \textit{set-erosion} strategy for verifying the safety of discrete-time stochastic systems. 
Our strategy states that to verify the safety of a stochastic system on a set, it is sufficient to verify the safety of an associated deterministic system on an eroded subset. 
The degree of erosion is quantified by the probabilistic bound on the distance between stochastic trajectories and their deterministic counterparts, termed stochastic trajectory gap. We provide a sharp probabilistic bound for this gap, enabling the set-erosion strategy to effectively reduce the stochastic safety verification problem to a deterministic one. Unlike martingale-based methods where designing a satisfying martingale is challenging, our method is easy to use and can be combined with any existing techniques for safety verification of deterministic systems. Moreover, our method, when combined with barrier functions, offers a significantly tighter result than existing methods when the failure probability bound is low and the time horizon is long.

\textit{Notations.} The set of positive integers is denoted by $\mathbb{N}_{+}$. We use $\|\cdot\|$ to denote $\ell_2$ induced norm. Given two sets $A,B\subseteq \R^n$, the Minkowski sum of the sets $A$ and $B$ is defined by $A\oplus B = \{x+y: x\in A,~ y\in B\}$, and the Minkowski difference is defined by $A\ominus B=(A^c\oplus(-B))^c$, where $A^c,B^c$ are the complements of $A,B$ and $-B=\{-y: y\in B\}$. We use $\mbE$ to denote expectations, $\mbP$ to denote the probability, $\mathcal{N}(\mu,\Sigma)$ to denote Gaussian distribution, $\mathcal{B}^n(r,y)$ to denote the ball $\{x\in\R^n: \|x-y\|\leq r\}$. For a random variable $X$, $X\sim G$ means $X$ is independent and identically drawn from the distribution $G$. We say $\alpha(\cdot): \R\to\R$ is an extended class $\mK$ function if $\alpha(0)=0$ and $\alpha(\cdot)$ is increasing on $\R$.

\section{Problem Statement} \label{sec: problem}
Consider the discrete-time stochastic system
\begin{equation}\label{sys: d-t ss}
     X_{t+1}=f(X_t,d_t,t)+w_t,
\end{equation}
where $X_t\in \R^n$ is the system state, $d_t\in\mD\subseteq\R^m$ is the bounded inputg, $w_t\in\R^n$ is the stochastic disturbance, and $f: \R^n\times\R^m\times\mathbb{N}_+\to\R^n$ is a smooth parameterized vector field. In this paper, we impose the Lipschitz nonlinearity condition on the system. 
\begin{assumption}\label{ass: Lipschitz f}
    At every time $t\geq0$, there exists $L_t\geq0$ such that $\|f(x,d,t) - f(y,d,t)\| \leq L_t\|x-y\|$ holds for every $x,y\in \R^n$ and every $d\in \R^m$.
\end{assumption}
We model $w_t$ as \textit{sub-Gaussian} disturbance, which includes a wide range of noise distributions such as Gaussian, uniform, and any zero-mean distributions with bounded support \cite[Section 2]{vershynin2018high}. 
\begin{definition}[sub-Gaussian] \label{def: subG}
    A random variable $X\in\R^n$ is said to be sub-Gaussian with variance proxy $\sigma^2$, denoted as $X\sim subG(\sigma^2)$, if $\mbE(X)=0$ and 
    for any $\ell$ on the unit sphere $\mS^{n-1}$,
$\mbE_X\left(e^{\lambda \innerp{\ell,X}}\right)\leq e^{\frac{\lambda^2\sigma^2}{2}},\mbox{for all } \lambda\in\R$.
\end{definition}
\smallskip
\begin{assumption}\label{ass: bounded sigma}
   For the discrete-time stochastic system~\eqref{sys: d-t ss}, $w_t\sim subG(\sigma_t^2)$ with some finite $\sigma_t>0$, $\forall t\geq0$. 
\end{assumption}

This paper aims to establish an effective safety verification method for the stochastic system \eqref{sys: d-t ss}. 
To formulate this problem, we first formalize the concept of safety for the deterministic systems \cite{ames2019control}.
Consider the deterministic system
\begin{equation}\label{sys: d-t ds}
    x_{t+1}=f(x_t,d_t,t),
\end{equation}
which can be treated as the noise-free version of the stochastic system \eqref{sys: d-t ss}. Given a terminal time $T\in\mathbb{N}_+$ and a safe set $\mC\subseteq\R^n$, we say the deterministic system \eqref{sys: d-t ds} starting from $\mX_0$ is \textit{safe} during $t\leq T$ if $\mX_0\subseteq\mC$ and
\begin{equation}\label{eq: det safety}
  x_0\in\mX_0 ~ \Rightarrow~ x_t\in\mC,~  \;\;\;\;\forall t\leq T,~ \forall d_t\in\mD.
\end{equation}

For the stochastic system \eqref{sys: d-t ss}, safety in the sense of \eqref{eq: det safety} can be restrictive. When $w_t$ is unbounded, $X_t$ is likely to be unbounded, thus any bounded set in $\R^n$ will be judged as unsafe. 
Even if $w_t$ is bounded, \eqref{eq: det safety} completely ignores the statistical property of the stochastic noise and requires $\mC$ to have enough robustness to the worst case of $w_t$, which rarely happens in applications. This usually leads to conservative safety guarantees. For these reasons, we focus on the stochastic safety with bounded failure probability \cite{steinhardt2012finite} to better capture the effect of the stochastic noise.
\begin{definition}
    Consider the stochastic system \eqref{sys: d-t ss} with the bounded set $\mD\subseteq\R^m$. Given a $\delta\in[0,1]$, an safe set $\mC\subset\mathbb{R}^n$, an initial configuration  $\mathcal{X}_0\subseteq\R^n$ and a terminal time $T$, the system is said to be \textit{safe with $1-\delta$ guarantee}
    during $t\leq T$ if $\mX_0\subseteq\mC$ and:
    \begin{equation} \label{eq: sto safety}
        X_0\in\mathcal{X}_0~ \Rightarrow~ \prob{X_t\in\mC,~ \forall t\leq T}\geq 1-\delta.
    \end{equation}
\end{definition}
\smallskip
With this definition, the stochastic safety verification problem that we seek to solve can be formalized as below. 
\begin{problem}[Stochastic Safety Verification]\label{prob1}
    Consider a stochastic system \eqref{sys: d-t ss} under Assumptions \ref{ass: Lipschitz f} and \ref{ass: bounded sigma}. Develop an effective strategy to verify its safety with $1-\delta$ guarantee during a finite horizon $t\leq T$. 
\end{problem}

\begin{figure}
\centering
\includegraphics[width =0.7\linewidth]{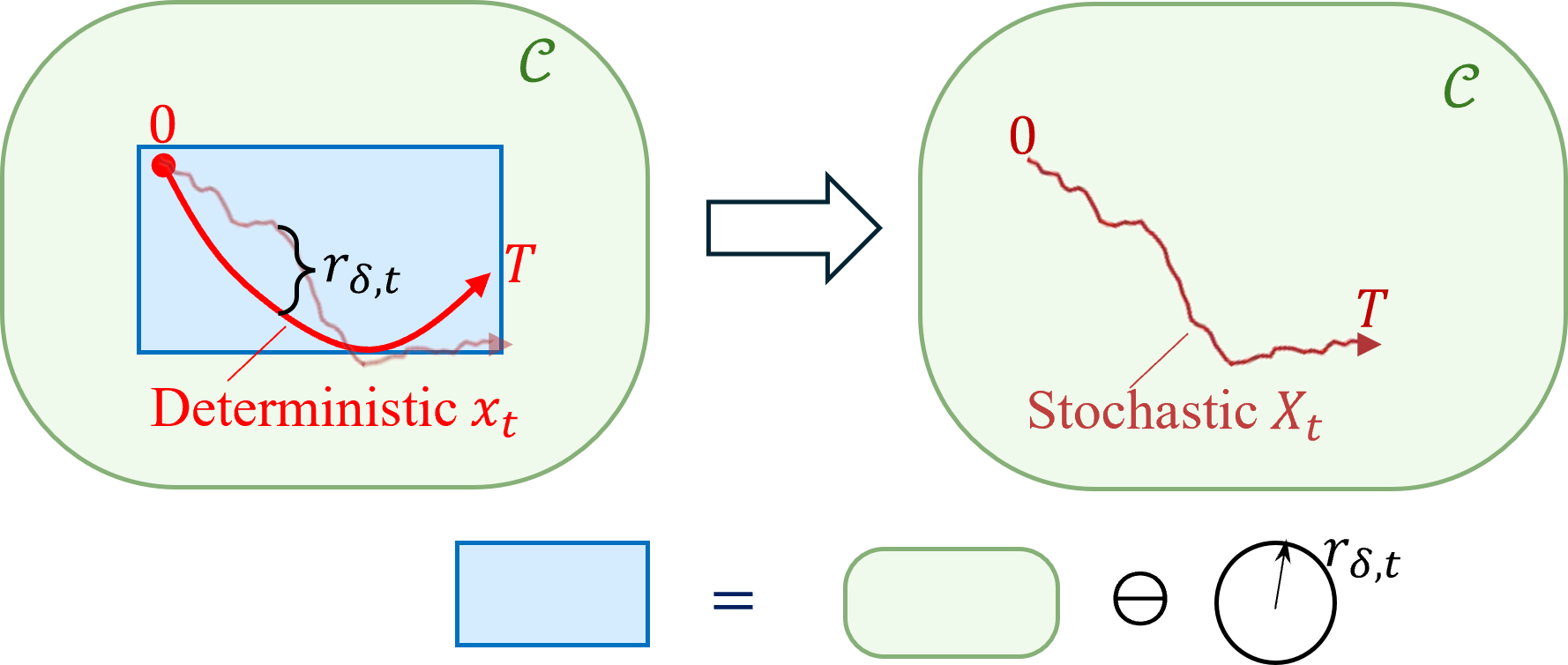}
\caption{An illustration of set-erosion strategy. Here $\mC$ in green is the safe set, and the blue area is the eroded subset $\mC\ominus\mathcal{B}^n\left(r_{\delta,t},0\right)$ with $r_{\delta,t}$ given in Theorem \ref{thm: set-erosion}-2). By Theorem \ref{thm: set-erosion}, if the deterministic trajectory stays in the blue area at any time, then the stochastic trajectory is safe on $\mC$ with $1-\delta$ guarantee. 
}
\label{fig: set-erosion}
\end{figure}

\section{Set-Erosion Strategy} \label{sec: erosion}
 Intuitively, the stochastic system \eqref{sys: d-t ss} is fluctuating around its associated deterministic system \eqref{sys: d-t ds} with high probability. 
 Given an safe set $\mC\subset\R^n$, if we erode/shrink $\mC$ from its boundary slightly to get a subset $\tmC\subset\mC$, and verify that the deterministic system \eqref{sys: d-t ds} is safe on $\tmC$, then the fluctuation of the stochastic trajectories would probably not exceed the ``robustness buffer'' $\mC\backslash\tmC$. Building on this intuition, we propose a strategy termed \textit{set-erosion} for stochastic safety verification. 
 This strategy can be viewed as the dual to the separation strategy for stochastic reachability analysis proposed in \cite{szy2024TAC}. 

For the associated systems \eqref{sys: d-t ss} and \eqref{sys: d-t ds}, we say $X_t$ and $x_t$ are \textit{associated} trajectories if they have the same initial state $X_0=x_0$ and the same input $d_t$. The fluctuation of the stochastic system \eqref{sys: d-t ss} around the deterministic system \eqref{sys: d-t ds} can be quantified by the distances among pairs of associated trajectories. The set-erosion strategy produces a sufficient condition on the safety of the stochastic system \eqref{sys: d-t ss} with $1-\delta$ guarantee, as formalized below.
\begin{thm}[Set-erosion strategy]\label{thm: set-erosion}
     Consider the stochastic system \eqref{sys: d-t ss} and its associated deterministic system \eqref{sys: d-t ds}. Given an initial set $\mathcal{X}_0$, a safe set $\mC\in\R^n$ and a terminal time $T$, suppose that there exists $r_{\delta,t}$ such that, for any trajectory $X_t$ of \eqref{sys: d-t ss} and its associated trajectory $x_t$ of \eqref{sys: d-t ds} starting from $\mathcal{X}_0$, 
     \begin{enumerate}
         \item \label{thm1-1} $\prob{\|X_t-x_t\|\leq r_{\delta,t},~\forall t\leq T}\geq 1-\delta$,
         \item \label{thm1-2} $ x_t\in\mC\ominus\mathcal{B}^n\left(r_{\delta,t},0\right)$, $\mbox{ for all } t\leq T$,
     \end{enumerate}
     then the system \eqref{sys: d-t ss} is safe with $1-\delta$ guarantee during $t\leq T$.
\end{thm}
\begin{proof}
    Let $X_t$ be any trajectory of \eqref{sys: d-t ss} associated with a trajectory $x_t$ of \eqref{sys: d-t ds}. Then, by Condition \ref{thm1-1} and the definition of the Minkowski sum,
\begin{align}\nonumber
   \prob{ X_t\in \{x_t\}\oplus\mathcal{B}^n(r_{\delta,t},0), \forall t\leq T}\geq 1-\delta.
\end{align}
By Condition \ref{thm1-2}, $\prob{ X_t\in \mC, \forall t\leq T}\geq 1-\delta$ follows. 
\end{proof}

An illustration of Theorem \ref{thm: set-erosion} is shown in Figure \ref{fig: set-erosion}. We term the $r_{\delta,t}$ in Theorem \ref{thm: set-erosion} as the probabilistic bound on \textit{stochastic trajectory gap}, as it represents the gap between stochastic trajectories and their deterministic counterpart over a time horizon. It quantifies the \textit{erosion depth} in the Minkowski difference $\mC\ominus\mathcal{B}^n\left(r_{\delta,t},0\right)$. Theorem \ref{thm: set-erosion} states that once a probabilistic bound $r_{\delta,t}$ of the stochastic trajectory gap is provided, then to verify the safety of the stochastic system on $\mC$,
 it suffices to verify the safety of its \textit{associated deterministic system} on the eroded subset $\mC\ominus\mathcal{B}^n\left(r_{\delta,t},0\right)$.  
 
The effectiveness of the set-erosion strategy relies on the tightness of $r_{\delta,t}$. If $r_{\delta,t}$ is too large, then $\mC\ominus\mathcal{B}^n\left(r_{\delta,t},0\right)$ can be very small or even empty, rendering conservative conditions. Therefore, it is crucial to establish a tight probabilistic bound $r_{\delta,t}$ for the stochastic trajectory gap.

\section{Probabilistic Bound on Stochastic Trajectory Gap} \label{sec: bound}
In this section, we present two approaches to probabilistically bounding the stochastic trajectory gap. 
The first one is based on a novel stochastic analysis technique developed in our previous work \cite{szy2024Auto}. The second one follows the idea of worst-case analysis and is presented for comparsion.
By comparing the two, we demonstrate that the former is always superior to the latter.

\subsection{Probabilistic Bound Based on Stochastic Deviation}

In~\cite{szy2024Auto}, we introduce the notion of \emph{stochastic deviation} as the distance $\|X_t-x_t\|$ between associated $X_t$ and $x_t$ at a single time $t$, and give a tight probabilistic bound on the stochastic deviation. 
\begin{proposition}[Stochastic deviation~\cite{szy2024Auto}]\label{prop: single-time}
     Consider the stochastic system \eqref{sys: d-t ss} and its associated deterministic system \eqref{sys: d-t ds} under Assumptions \ref{ass: Lipschitz f} and \ref{ass: bounded sigma}. Let $X_t$ be the trajectory of \eqref{sys: d-t ss} and $x_t$ be the associated trajectory of \eqref{sys: d-t ds} Then, given $t\geq0$, for any $\delta\in(0,1)$ and $\varepsilon\in(0,1)$,
     \begin{equation}\label{eq: prop1}
         \|X_t-x_t\|\leq\sqrt{\Psi_t(\varepsilon_1n+\varepsilon_2\log(1/\delta))}
     \end{equation}
     holds with probability at least $1-\delta$, where 
        \begin{equation}\label{eq: Psi val}
         \Psi_t=\psi_{t-1}\sum_{k=0}^{t-1}\sigma_{k}^2\psi_k^{-1},\quad 
             \psi_t=\prod_{k=0}^{t}L_k^{2},
     \end{equation}
      \begin{equation}\label{eq: epsilon val}
        \varepsilon_1=\frac{2\log(1+2/\varepsilon)}{(1-\varepsilon)^2},~ \varepsilon_2=\frac{2}{(1-\varepsilon)^2}.
    \end{equation}
\end{proposition}
\begin{remark}\label{remark: parameter}
    In general, the choice of $\varepsilon_1$,$\varepsilon_2$ based on~\eqref{eq: epsilon val} is not necessarily optimal. For instance, when the state dimension $n=1$, one can choose $\varepsilon_1=2\log2$ and $\varepsilon_2=2$ by Hoeffding's Inequality \cite[Chapter 1]{rigollet2023high} for a tighter bound. 
\end{remark}

Based on Proposition \ref{prop: single-time}, we establish a probabilistic bound on the stochastic trajectory gap over a finite horizon. 

\begin{thm}[Stochastic trajectory gap]\label{thm: union bound}
    Consider the stochastic system \eqref{sys: d-t ss} and its associated deterministic system \eqref{sys: d-t ds} under Assumption \ref{ass: Lipschitz f} and \ref{ass: bounded sigma}. Let $X_t$ be the trajectory of \eqref{sys: d-t ss} and $x_t$ be the associated trajectory of \eqref{sys: d-t ds} with the same initial state $x_0\in\mX_0$ and input $d_t\in\mD$ on $t\leq T$. For any given $\delta\in(0,1]$ and desired $\varepsilon\in(0,1)$, define
    \begin{equation}\label{eq: r=}
        r_{\delta,t}=
      \sqrt{\Psi_t(\varepsilon_1n+\varepsilon_2\log(\tfrac{T}{\delta})},
    \end{equation}
   where $\Psi_t$ is as in \eqref{eq: Psi val} and $\varepsilon_1,\varepsilon_2$ are as in \eqref{eq: epsilon val}. Then 
   \begin{equation}\label{eq: thm1}
\mathbb{P}\left(\|X_t-x_t\|\leq r_{\delta,t}, ~\forall t\leq T\right)\geq 1-\delta.
   \end{equation}
\end{thm}

\begin{proof}
     Given $t\in[0,T]$, Proposition \ref{prop: single-time} implies that for any associated $X_t$ and $x_t$ at time $t$, it holds that 
    \begin{equation}\label{eq: one-point}
        \mathbb{P}\left(\|X_t-x_t\|> \sqrt{\Psi_t(\varepsilon_1n+\varepsilon_2\log(\tfrac{T}{\delta})}\right)\leq \tfrac{\delta}{T},
    \end{equation}
    where $\Psi_t$ is as in \eqref{eq: Psi val} and $\varepsilon_1,\varepsilon_2$ are as in \eqref{eq: epsilon val}. Define $r_{\delta,t}=
      \sqrt{\Psi_t(\varepsilon_1n+\varepsilon_2\log(\tfrac{T}{\delta})}$. Applying union bound inequality to \eqref{eq: one-point} over $t=1,\dots,T$ yields
    \begin{equation*}\label{eq: union bound ineq}
        \begin{split}
            \mathbb{P}(\bigcap_{t=1}^T\|X_t-x_t\|\leq r_{\delta,t})
            =&1-\mathbb{P}(\bigcup_{t=1}^T\|X_t-x_t\|> r_{\delta,t}) \\
            \geq &1-\sum_{t=1}^T \frac{\delta}{T}=1-\delta,
        \end{split}
    \end{equation*}
    which completes the proof.
\end{proof}

For the special case when $L_t\equiv L$ and $\sigma_t\equiv\sigma$ with some $L,\sigma>0$ for every $t\leq T$, the expression for $r_{\delta,t}$ in Theorem \ref{thm: union bound} can be simplified as follows:
    \begin{equation}\label{eq: invariant L}
       r_{\delta,t}=\sqrt{\tfrac{\sigma^2(L^{2t}-1)}{L^2-1}(\varepsilon_1n+\varepsilon_2\log(\tfrac{T}{\delta})}.
    \end{equation}

Compared to the single-time probabilistic bound \eqref{eq: prop1} in Proposition \ref{prop: single-time}, applying the union bound inequality only leads an additional $\mO(\sqrt{\log T})$ term in $r_{\delta,t}$ derived in Theorem \ref{thm: union bound}, which scales logarithmically with $T$. Moreover, the bound \eqref{eq: prop1} is proved to be tight for the stochastic system \eqref{sys: d-t ss}, and is exact for linear systems \cite[Section 4.4]{szy2024Auto}. Therefore, $r_{\delta,t}$ in \eqref{eq: r=} is a sharp probabilistic bound on stochastic trajectory gap. 
A comparison with some existing methods is displayed in Section \ref{sec: case}, showing that our result is much tighter.

\subsection{Probabilistic Bound by Worst-Case Analysis}

The worst-case analysis is a commonly-used method for safety verification when the disturbance is bounded~\cite{prajna2007framework,novara2013direct}.
It can also be applied to stochastic systems to estimate stochastic trajectory gap under any sub-Gaussian stochastic disturbance $w_t$. This is achieved by viewing $w_t$ as a bounded disturbance with high probability. However, the result is more conservative than that in Theorem \ref{thm: union bound}. 

By the norm concentration properties of sub-Gaussian random variables \cite[Chapter 1.4]{rigollet2023high} and the union bound inequality, the bound 
    \begin{equation}\label{eq: b val}
     b_t\geq \sqrt{\sigma_t^2(\varepsilon_1n+\varepsilon_2\log(\tfrac{T}{\delta}))}.
\end{equation}
for all $t\leq T$ 
ensures that $\prob{\|w_t\|\leq b_t,~\forall t\leq T}\geq1-\delta$.
A worst-case probabilistic bound on $\|X_t-x_t\|$ can be established by assuming this bound \eqref{eq: b val}. More specifically, by the local Lipschitz assumption and the triangular inequality,
\begin{equation*}
\begin{split}
    \|X_{t+1}-x_{t+1}\|\leq&\|f(X_t,d_t,t)-f(x_t,d_t,t)\|+\|w_t\| \\
    \leq& L_t\|X_t-x_t\|+b_t 
\end{split}
\end{equation*}
for all $t<T$ with probability at least $1-\delta$.
It follows that
\begin{equation}\label{eq: xt-xt worst}
    \|X_t-x_t\|\leq \sqrt{\psi_{t-1}}\sum_{k=0}^{t-1}b_{k}\sqrt{\psi_k^{-1}},
\end{equation}
where $\psi_t$ is as in \eqref{eq: Psi val}. Plugging \eqref{eq: b val} into \eqref{eq: xt-xt worst}, we conclude
\begin{equation}\label{eq: sd by worst}
\begin{split}
     &\|X_t-x_t\| \\
     \leq&\sqrt{\psi_{t-1}}\sum_{k=0}^{t-1}\sigma_{k}\sqrt{\psi_k^{-1}(\varepsilon_1n+\varepsilon_2\log\tfrac{T}{\delta})}, ~\forall t\leq T
\end{split}
\end{equation}
holds with probability at least $1-\delta$. 

This bound \eqref{eq: sd by worst} derived using worst-case analysis is substantially more conservative than that in Theorem \ref{thm: union bound}. Indeed, since $\sqrt{\Psi_t}\leq \sqrt{\psi_{t-1}}\sum_{k=0}^{t-1}\sigma_{k}\sqrt{\psi_k^{-1}}$ by \eqref{eq: Psi val}, \eqref{eq: sd by worst} is always worse than \eqref{eq: r=}-\eqref{eq: thm1}. To see more clearly the gap, consider the case when $L_t\equiv L$ and $\sigma_t\equiv\sigma$. In this case, \eqref{eq: sd by worst} reduces to $\|X_t-x_t\|\leq \frac{L^t-1}{L-1}\sqrt{\sigma^2(\varepsilon_1n+\varepsilon_2\log\frac{t}{\delta})}$, which is much worse than \eqref{eq: invariant L}, especially when $L\approx 1$ or $\ge 1$. 

\section{Set Erosion with Barrier Functions}\label{sec: barrier}


By combining the set-erosion strategy in Theorem \ref{thm: set-erosion} with the sharp probabilistic bound proposed on the stochastic trajectory gap developed in Theorem \ref{thm: union bound}, the stochastic system \eqref{sys: d-t ss} starting from $\mX_0\subseteq\mC$ is safe with $1-\delta$ guarantee, if 
\begin{equation}\label{eq: set ero}
    \begin{split}
       x_0\in \mathcal{X}_0 \Rightarrow x_t\in \mC\ominus\mathcal{B}^n\left(r_{\delta,t},0\right),~ r_{\delta,t}~ \mbox{is as \eqref{eq: r=}}, ~\forall t\leq T.
    \end{split}
\end{equation}
The formulation~\eqref{eq: set ero} converts the stochastic safety verification problem into a deterministic safety verification on a time-varying set.  Thus, the new formulation~\eqref{eq: set ero} offers tremendous flexibility to Problem~\ref{prob1} as one can leverage any deterministic safety verification methods. 
A large number of existing approaches for safety verification of deterministic systems are based upon barrier functions~\cite{prajna2007framework,agrawal2017discrete}.
%
%
In this section, we focus on two types of commonly-used barrier functions, namely reciprocal barrier function and  exponential barrier function. By combining these notions of barrier function with our set-erosion strategy, we provide two efficient frameworks for safety verification of discrete-time stochastic systems.


\subsubsection{Reciprocal Barrier Function}
 We first extend the reciprocal barrier function proposed in \cite{agrawal2017discrete} to time-varying reciprocal barrier function (TV-RBF) as follows. 
\begin{definition}[TV-RBF]\label{def: DTTVRBF}
    Consider the deterministic system \eqref{sys: d-t ds} with $d_t\in\mD$. Given a terminal time $T$ and a time-varying set $\tmC_t=\{x\in\R^n:~ h(x,t)\geq0\}$ with smooth function $h(x,t)$, a function $B(x,t):\tmC_t\times\mathbb{N}\to\R$ is a discrete-time \emph{time-varying reciprocal barrier function (TV-RBF)} on $\tmC_t$ if there exist extended class $\mK$ functions $\alpha_1(\cdot)$, $\alpha_2(\cdot)$, $\alpha_3(\cdot)$ such that, for all $x\in\tmC_t$ and all $t\leq T$:
\begin{enumerate}
    \item\label{eq: rbf_1} $\frac{1}{\alpha_1(h(x,t))}\leq B(x,t)\leq \frac{1}{\alpha_2(h(x,t))}$, and 
    \smallskip
    \item\label{eq: rbf_2} $B(f(x,d,t),t+1)-B(x,t)\leq \alpha_3(h(x,t))$, $\forall d\in\mD$.
\end{enumerate}    
\end{definition}
\smallskip

When $\mathcal{X}_0\subseteq \tmC_t\subseteq\mC\ominus\mathcal{B}^n\left(r_{\delta,t},0\right)$, the existence of $B(x,t)$ on $\tmC_t$ given as Definition \ref{def: DTTVRBF} guarantees that the deterministic system \eqref{sys: d-t ds} is safe. Therefore, the set-erosion strategy in~\eqref{eq: set ero} implies that the stochastic system \eqref{sys: d-t ss} is safe with $1-\delta$ guarantee. This result is formalized as follows.
\begin{proposition}(Safety using TV-RBF)\label{prop: rbf}
    Consider the stochastic system \eqref{sys: d-t ss} with the initial configuration $\mX_0\subseteq\R^n$ disturbance set $\mD\subseteq\R^m$ and its associated deterministic system \eqref{sys: d-t ds} under Assumptions \ref{ass: Lipschitz f} and \ref{ass: bounded sigma}. Given a safe set $\mC\subseteq\R^n$ and terminal time $T$, define $r_{\delta,t}$ as \eqref{eq: r=}. If $X_0\in\mathcal{X}_0$ and there exist a subset $\tmC_t$ such that $\mX_0\subseteq\tmC_t \subseteq\mC\ominus\mathcal{B}^n\left(r_{\delta,t},0\right)$ and a TV-RBF $B(x,t)$ as defined in Definition \ref{def: DTTVEBF} on $\tmC_t$ for the deterministic system~\eqref{sys: d-t ds}, then the stochastic system \eqref{sys: d-t ss} is safe with $1-\delta$ guarantee on $\mC$. 
\end{proposition}
\begin{proof}
    To begin with, $x_0\in\mathcal{X}_0\subseteq\tmC_0$. For any $t\in\mathbb{N}_+$, suppose that $x_t\in\tmC_t$, then by \cite[Proposition 3]{agrawal2017discrete}, $\frac{1}{B(x_{t+1},t+1)}\geq0$ holds.
    Since the inverse of an extended class $\mK$ function is again an extended class $\mK$ function,
    \begin{equation*}
        h(x_{t+1},t+1)\geq\alpha_1^{-1}\left(\frac{1}{B(x_{t+1},t+1)}\right)\geq \alpha_1^{-1}(0)=0,
    \end{equation*}
    which implies that $x_{t+1}\in\tmC_{t+1}$. Using induction, one can prove that $x_t\in\tmC_t$ for any $t\leq T$, and thus the deterministic system \eqref{sys: d-t ds} remains on $\tmC_t$. Since $\tmC_t\subseteq\mC\ominus\mathcal{B}^n\left(r_{\delta,t},0\right)$, the deterministic trajectory $x_t$ satisfies the set-erosion strategy in \eqref{eq: set ero}, which suffices to verify the safety of the stochastic system \eqref{sys: d-t ss} with $1-\delta$ guarantee on $\mathcal{C}$.
\end{proof}
\bigskip
\subsubsection{Exponential Barrier Function}
One potential issue of the reciprocal barrier function is that it tends to  infinity as its argument approaches the boundary of the safe set \cite{ames2016control}. To address this issue, the notion of exponential barrier function has been proposed in the literature~\cite{agrawal2017discrete}. Given a time-varying set $\tmC_t\subseteq \mathbb{R}^n$, we generalize this notion and introduce the discrete-time time-varying exponential barrier function (TV-EBF) for the set $\tmC_t$. 
\begin{definition}[TV-EBF]\label{def: DTTVEBF}
    Consider the deterministic system \eqref{sys: d-t ds} with $d_t\in\mD$, $\mD\subseteq\R^m$. Given a terminal time $T$, if there exists a smooth function $h(x,t):\R^n\times\mathbb{N}\to\R$ such that for any $t\leq T$:
         \begin{enumerate}
         \item $\tmC_t=\{x\in\R^n:~ h(x,t)\geq0\}$, and 
         \item there exists $\gamma\in(0,1]$ such that, for all $d\in\mD$,
         \begin{equation*}
             h(f(x,d,t),t+1)\geq (1-\gamma) h(x,t),~ \mbox{ for all } x\in\tmC_t.
         \end{equation*}
     \end{enumerate}
     Then $h(x,t)$ is a \emph{discrete-time time-varying exponential barrier function (TV-EBF)} for the set $\tmC_t$.
\end{definition}

Similar to TV-RBF, by choosing $\mathcal{X}_0\subseteq \tmC_t\subseteq \mC\ominus\mathcal{B}^n\left(r_{\delta,t},0\right)$, the existence of a TV-EBF for the set $\tmC_t$ for the deterministic system~\eqref{sys: d-t ds} can certify the safety of the stochastic system \eqref{sys: d-t ss} with $1-\delta$ guarantee.
\begin{proposition}(Safety using TV-EBF)\label{prop: ebf}
    Consider the stochastic system \eqref{sys: d-t ss} with the initial configuration $\mathcal{X}_0\subseteq \R^n$ and the disturbance set $\mD\subseteq\R^m$ and its associated deterministic system \eqref{sys: d-t ds} under Assumptions \ref{ass: Lipschitz f} and \ref{ass: bounded sigma}. Given a safe set $\mC\subseteq\R^n$ 
     and terminal time $T$, define $r_{\delta,t}$ as \eqref{eq: r=}. If $X_0\in\mX_0$ and there exist a $\tmC_t$ such that $\mathcal{X}_0\subseteq \tmC_t\subseteq \mC\ominus\mathcal{B}^n\left(r_{\delta,t},0\right)$ and a TV-EBF $h(x,t)$ as defined in Definition \ref{def: DTTVEBF} on $\tmC_t$ for the deterministic system~\eqref{sys: d-t ds}, then the stochastic system \eqref{sys: d-t ss} is safe with $1-\delta$ guarantee on $\mC$. 
\end{proposition}
\begin{proof}
   Clearly it holds that $h(x_t,t)\geq (1-\gamma)^th(x_0,0)$ for any $t\leq T$.
   %
   Since $x_0\in\mathcal{X}_0\subseteq\tilde\mC_0$, this implies that $h(x_t,t) \ge 0$, for every $t\le T$, 
   Therefore $x_t\in\tilde\mC_t$, for every $t\le T$. Since $\tmC_t\subseteq\mC\ominus\mathcal{B}^n\left(r_{\delta,t},0\right)$, $x_t$ satisfies the set-erosion strategy in \eqref{eq: set ero}, which suffices to show that the stochastic system \eqref{sys: d-t ss} is safe with $1-\delta$ guarantee on $\mathcal{C}$. 
\end{proof}

Propositions \ref{prop: rbf} and \ref{prop: ebf} provide two stochastic safety verification schemes based on deterministic barrier certifications, whose efficiency depends on the calculation of the Minkowski difference $\mC\ominus\mathcal{B}^n\left(r_{\delta,t},0\right)$. Notice that $\mC\ominus\mathcal{B}^n\left(r_{\delta,t},0\right)=\mC^c\oplus\mathcal{B}^n\left(r_{\delta,t},0\right)$, where the complement $\mC^c$ of $\mC$ is treated as the unsafe set. When $\mC^c$ is the union of convex sets such as the union of ellipsoids or polyhedral, $\mC^c\oplus\mathcal{B}^n\left(r_{\delta,t},0\right)$ can be efficiently computed \cite{weibel2007minkowski}.

\section{Case Studies}\label{sec: case}
In this section, we present two examples to validate the proposed safety verification method. In the first example, we verify safety of a linear scalar stochastic system on a finite interval. In the second example, we verify safety of unicycle model of a vehicle with a stabilizing feedback controller. 

\subsection{Safety of Linear Systems over an Interval}
As the first experiment, we consider the following linear stochastic system
\begin{equation}\label{sys: lin}
    X_{t+1}=0.99X_t+w_t,~ X_0=0,
\end{equation}
where $X_t\in\R$ and $w_t\sim \mathcal{N}(0,10^{-3})$. This linear system satisfies Assumption \ref{ass: Lipschitz f} with $L_t\equiv L=0.99$ and Assumption \ref{ass: bounded sigma} with $\sigma_t^2\equiv\sigma^2=10^{-3}$. The probabilistic bound $r_{\delta,t}$ on stochastic trajectory gap can be calculated by \eqref{eq: invariant L} with $\varepsilon_1=2\log2$ and $\varepsilon_2=2$ by Remark \ref{remark: parameter}.

The task is to verify the safety of the linear system \eqref{sys: lin} with $1-\delta$ guarantee on the interval $\mC=\{x\in\R: |x|\leq R\}$ during $t\leq T$. Notice that by fixing $X_0=0$, the associated deterministic trajectory of the system \eqref{sys: lin} is $x_t\equiv0$, and $\ballone{R,0}\ominus\ballone{r_{\delta,t},0}=\ballone{R-r_{\delta,t},0}$. Therefore, by our set-erosion strategy in~\eqref{eq: set ero}, it is enough to verify whether $R\geq r_{\delta,t}$ for any $t\leq T$. Since $r_{\delta,t}$ calculated by \eqref{eq: invariant L} is increasing with $t$, we conclude that it suffices to verify if $R\geq r_{\delta,T}$, which is equivalent to:
\begin{equation}\label{eq: delta>=}
    \delta\geq T\cdot\exp\left\{-\left(\frac{R^2(L^2-1)}{2\sigma^2(L^{2T}-1)}-\log2\right)\right\}.
\end{equation}


The right-hand side of \eqref{eq: delta>=} is the lowest probability $\underline{\delta}$ that our strategy \eqref{eq: set ero} cannot guarantee safety for the linear system \eqref{sys: lin}. When $\underline{\delta}\geq1$, it means that the radius $R$ of $\mC$ is so small that the system \eqref{sys: lin} is considered unsafe on $\mC$. In such a scenario we set $\delta=1$. Figure \ref{fig: Lin delta-R} shows the relationship between $\underline{\delta}$ and $R$ determined by \eqref{eq: delta>=}. Our result is compared with the curve given by the main result of \cite{cosner2024bounding}, which is the best existing result to the best of our knowledge, and the simulated result given by Monte-Carlo approximations with $3\times10^6$ sampled trajectories. When $R$ is very small, our strategy directly implies that the system is unsafe on $\mC$, as suggested by the simulated result. When $R$ gets larger, our strategy offers a result close to the simulated result and significantly sharper than existing methods.

\begin{figure}[t]
	\centering
  \includegraphics[width =0.48\linewidth]{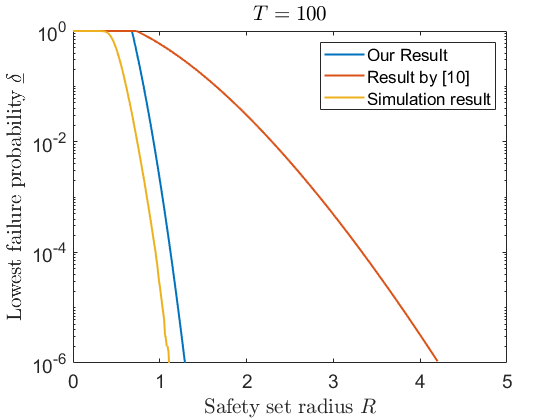}
  \includegraphics[width =0.48\linewidth]{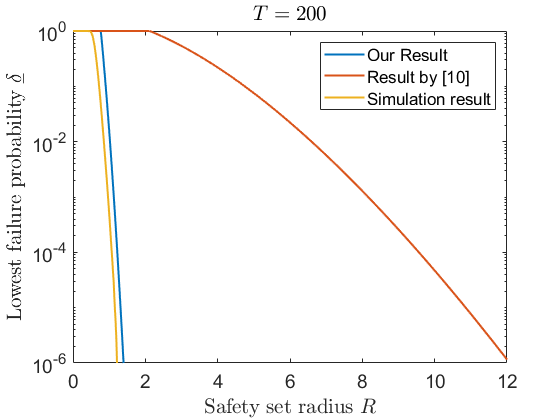}
	\caption{The probability $\underline{\delta}$ (the lower the safer) that a strategy cannot guarantee safety on the centered ball with radius $R$ during $t\leq T$. The blue curve is given by our strategy \eqref{eq: delta>=}, the yellow curve is the simulated result from $3\times10^6$ sampled trajectories, and the red curve is the baseline we choose, given by the main result of \cite[Corollary 1]{cosner2024bounding}. \textbf{Left:} $T=100$. \textbf{Right:} $T=200$. 
 }
	\label{fig: Lin delta-R}
\end{figure} 


\subsection{Nonlinear Unicycle System}
Next, we consider a unicycle moving on a $2$-dimensional plane with obstacles. The unsafe region $\mC_u=\{(p_x-1.5)^2+(p_y-3.5)^2\geq0.9^2\}\cap \{(p_x+0.5)^2+(p_y-2)^2\geq0.72^2\}\cap \{(p_x-6.2)^2+(p_y-0.7)^2\geq0.75^2\}$ is the union of red obstacles shown in Figure~\ref{fig: unicycle} and the safe region is $\mC=\R^n/\mC_u$. The discrete-time system model is given as
\begin{equation}\label{sys: uni}
    \begin{split}
        X_{t+1} &= X_t + \eta\begin{bmatrix}
    v_t(X_t) \cos(\theta_t)\\
    v_t(X_t) \sin(\theta_t)\\
    \omega_t(X_t) + d_t
\end{bmatrix} + w_t \\
&\defeq f(X_t,d_t)+w_t,
    \end{split}
\end{equation}
where $X_t = \begin{bmatrix}p_{x,t} & p_{y,t} & \theta_t\end{bmatrix}^{\top}$ is the state of the vehicle, $(p_{x,t},p_{y,t})$ is the position of the center of mass of the vehicle in the plane, $\theta_t$ is the heading angle of the vehicle, $v_t(X_t)$ is the velocity of the center of mass, $\omega_t(X_t)$ is the angular velocity of the vehicle, $\eta$ is the discretization step size, $d_t$ is a bounded disturbance on the angular velocity, and $w_t$ is the stochastic disturbance on the model. 
In this experiment, we suppose that $|d_t|\leq0.1$, $w_t~\sim \sqrt{\eta}\cdot\mathcal{N}(0,0.01)$, $\eta=0.01$. $v_t(X_t)$ and $\omega_t(X_t)$ are designed as the feedback controllers proposed in~\cite{MA-GC-AB-AB:95}. The task of the unicycle is to reach the origin point while avoiding the obstacles under both $d_t$ and $w_t$. The details of controller design can be found in \cite[Section VIII]{szy2024TAC}.

Our goal is to verify the safety of the stochastic system \eqref{sys: uni} with $1-\delta$ guarantee on through set erosion strategy. We set  $\delta=10^{-4}$ and the initial state set $\mX_0=\begin{bmatrix}5 & 5 & -\frac{\pi}{3} \end{bmatrix}^{\top}\pm0.1$. The probabilistic bound $r_{\delta,t}$ is calculated as \eqref{eq: r=} with $\varepsilon=1/16$, and $L_t$ estimated by the methods proposed in \cite{chuchu2017simulation}. Define $r_m=\max_{t\leq T}r_{\delta,t}$, then by the set erosion strategy \eqref{eq: set ero}, it suffices to verify whether for the associated deterministic system of \eqref{sys: uni}, $x_t\in\mC\ominus \mathcal{B}^n\left(r_m,0\right)$ holds for any $x_t$ starting from $x_0\in\mX_0$ and under any $|d_t|\leq0.1$. We use barrier certification to verify this condition, where Proposition \ref{prop: rbf} and Proposition \ref{prop: ebf} reduce to typical barrier certification for forward-invariant condition \cite{ames2019control}. The tool we use is FOSSIL developed by \cite{abate2021fossil}. Based on the experiment setting, this program returns \texttt{``Found a valid BARRIER certificate''}, implying that the safety of the system \eqref{sys: uni} with $1-\delta$ guarantee on the zero-superlevel set of this barrier function is verified.

To visualize the set-erosion strategy, we simulate 5000 independent trajectories of the associated deterministic system from $x_0\in\mX_0$ with $T=100,500$ separately. The results are shown in the left column of Figure \ref{fig: unicycle}. The areas in yellow are the eroded parts $\mC\backslash(\mC\ominus\mathcal{B}^n\left(r_m,0\right))$ of the $\mC$. It is clear that all the deterministic trajectories have no intersections with the yellow areas. Meanwhile, to validate the effectiveness of our strategy, we sample 20000 independent trajectories of the stochastic system \eqref{sys: uni} from $X_0\in\mX_0$ during $t\leq T$. It is clear that all the stochastic trajectories successfully avoid all the obstacles, satisfying our safety verification strategy.

\begin{figure}[t]
	\centering
  \includegraphics[width =0.48\linewidth]{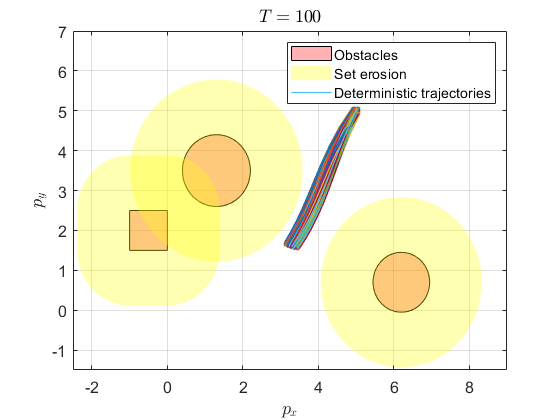}
    \includegraphics[width =0.48\linewidth]{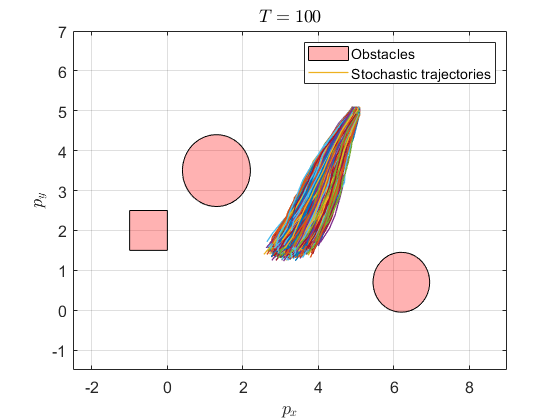}
    \includegraphics[width =0.48\linewidth]{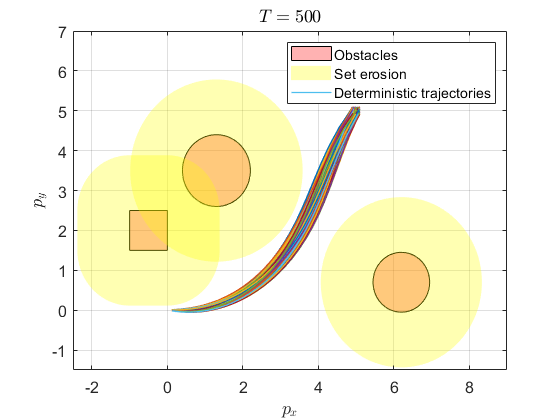}
    \includegraphics[width =0.48\linewidth]{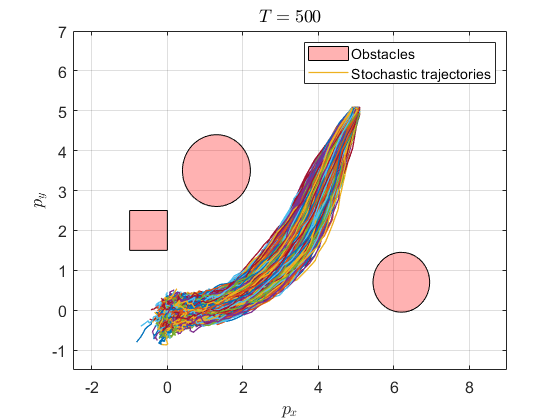}
	\caption{Stochastic safety verification of the unicycle system \eqref{sys: uni} with $1-10^{-4}$ failure guarantee. \textbf{Left:} Stochastic safety verification using the set-erosion strategy at the terminal time $T=100,500$. The red shapes are obstacles. The yellow areas are the eroded part $\mC\backslash(\mC\ominus\mathcal{B}^n\left(r_m,0\right))$. Each curve is an independent trajectory of the associated deterministic unicycle system with different $d_t$. \textbf{Right:} Simulation of the stochastic trajectories. Each curve is an independently sampled trajectory of the stochastic system \eqref{sys: uni} during $t\leq T$, $T=100,500$.}
	\label{fig: unicycle}
\end{figure} 

\section{Conclusion} \label{sec: conclusion}
We propose a general approach called set-erosion strategy for safety verification of discrete-time stochastic systems with sub-Gaussian disturbances. 
 Our set-erosion strategy reduces the problem of safety verification of discrete-time stochastic systems into the safety verification of an associated deterministic system on an eroded subset of the safe set. 
Based on our results in \cite{szy2024Auto}, we provide a sharp probabilistic bound on the depth of this erosion. 
This approach brings huge flexibility to safety verification of stochastic systems as any deterministic safety verification methods can be used to ensure safety on the eroded subset of the safe set.
In particular, we consider two types of barrier functions for safety verification of deterministic systems and leveraged them to obtain efficient stochastic safety verification schemes. 


\bibliographystyle{ieeetr}
\bibliography{main}    
\end{document}